\documentclass[12pt]{amsart}
\usepackage{graphicx}
\usepackage{verbatim}
\usepackage{lineno}

\newtheorem{theorem}{Theorem}%[section]

\newtheorem{proposition}[theorem]{Proposition}

\def\IR{{\rm I\!R}} 
\def\IC{\hbox{\rm C\kern-.43em
       \vrule depth 0ex height 1.4ex width .05em\kern.41em}}
\def\IQ{\hbox{\rm Q\kern-.43em
       \vrule depth 0ex height 1.4ex width .05em\kern.41em}}
   \def\bd{{\bf d}}
   \def\bF{{\bf F}}
\def\bx{{\bf x}} \def\by{{\bf y}}  
   
   \def\cS{{\mathcal Sinks}}
  \def\cI{{\mathcal I}}

\def\bS{{\bf S}}

\def\bA{{\bf A}}

\def\diag{{\rm diag}\,}
\def\qed{\hfill\vbox{\hrule width 6 pt\hbox{\vrule height 6 pt width 6 
pt}}}

\usepackage[numbers,sort&compress]{natbib} % sort and compress orders the numbers correctly

\begin{document}

\bibliographystyle{unsrtnat} % JLS: changed this get refs in order of appearance (saves space!)

\title[Evolution of dispersal in periodic environments]{Evolution of unconditional 
dispersal in periodic environments} 
\author[S. J. Schreiber]{Sebastian J. Schreiber}
\address{Department of Evolution and Ecology and Center for Population Biology, University of California, Davis, CA 95616}
\email{sschreiber@ucdavis.edu}
\author[C.K. Li]{Chi-Kwong Li}
\address{Department of Mathematics, College of William and Mary, Williamsburg, VA}
\email{ckli@math.wm.edu}
\maketitle

\begin{abstract}
\linenumbers
\baselineskip 20pt
Organisms modulate their fitness in heterogeneous  environments by dispersing. Prior work shows that there is selection against ``unconditional'' dispersal in spatially heterogeneous environments. ``Unconditional'' means individuals disperse at a rate independent of their location. We prove that if within-patch fitness varies spatially and between two values temporally, then there is selection for unconditional dispersal: any evolutionarily stable strategy (ESS) or evolutionarily stable coalition (ESC) includes a dispersive phenotype. Moreover, at this ESS or ESC, there is at least one sink patch (i.e. geometric mean of fitness less than one) and no sources patches (i.e. geometric mean of fitness greater than one). These results coupled with simulations suggest that spatial-temporal heterogeneity due to abiotic forcing result in either an ESS with a dispersive phenotype or an ESC with sedentary and dispersive phenotypes. In contrast, spatial-temporal heterogeneity due to biotic interactions can select for higher dispersal rates that ultimately spatially synchronize population dynamics. 
\end{abstract}

%\newpage
\section{Introduction}
\linenumbers
\baselineskip 20pt
Plants and animals often live in landscapes where environmental
conditions vary in space and time. These
environmental conditions may include abiotic factors such as light,
space, and nutrient availability or biotic factors such as prey,
competitors, and predators. Since the fecundity and survivorship of
an individual depends on these factors, an organism may decrease or
increase its fitness by dispersing across the environment.  Understanding how
natural selection acts on 
dispersal in heterogeneous environments has been the focus of much theoretical and empirical 
work~\cite{fretwell-lucas-70,hamilton-may-77,hastings-83,levin-etal-84, holt-85,pulliam-88,cohen-levin-91, mcpeek-holt-92, holt-96, doebeli-ruxton-97,diffendorfer-98,dockery-etal-98, parvinen-99, mathias-etal-01,friedenberg-03,hutson-etal-03,siap-06, cantrell-etal-06,cantrell-etal-07,dercole-etal-07, amnat-09a}. 

In spatially heterogeneous environments, the general consensus is that there is 
selection against unconditional dispersal~\cite{hastings-83,mcpeek-holt-92,dockery-etal-98,parvinen-99,siap-06}. Here, unconditional refers to the assumption that 
individuals disperse at a rate independent of their location. Using reaction-diffusion equations, Dockery et al.~\cite{dockery-etal-98} proved that for two 
competing populations only differing in their diffusion constant, the 
population with the larger diffusion constant is excluded. In \cite{parvinen-99,siap-06}, 
similar results were proven for populations with non-overlapping generations 
living in a patchy environment.  Alternatively, in temporally but not spatially 
heterogenous environments, Hutson et al.~\cite{hutson-etal-03} proved that 
dispersal rates are a selectively neutral trait for reaction-diffusion models 
and, thereby, confirmed the numerical observations of McPeek and 
Holt~\cite{mcpeek-holt-92} for discrete-time, two-patch models.  These results 
imply that the slightest cost of dispersal would result in selection against 
dispersal in purely temporally heterogenous environments. 

When there is a mixture of spatial and temporal heterogeneity, the interaction 
between competing dispersal phenotypes becomes more subtle. Numerically 
simulating discrete-time, two-patch models, McPeek and Holt~\cite{mcpeek-holt-92}, Parvinen~\cite{parvinen-99}, and Mathias et al. ~\cite{mathias-etal-01} found that more dispersive phenotypes could displace more sedentary  phenotypes for certain forms of spatial-temporal heterogeneity, while evolutionarily stable coalitions of sedentary and dispersal phenotypes are possible for other forms  of spatial-temporal heterogeneity.  Hutson et al. ~\cite{hutson-etal-03} proved that similar phenomena could occur for reaction diffusion equations. However,  analytical criteria distinguishing these outcomes remain elusive. 

In this article, we consider the evolution of dispersal for a general class of 
multi-patch difference equations varying periodically in time. This periodic 
variation can be either due to biotic interactions or abiotic forcing.  Our 
main goals are to analytically identify potential evolutionarily stable strategies or 
coalitions for dispersal, characterize the spatial-temporal patterns 
of fitness generated by populations playing these dispersal strategies, and use our results to compare evolutionary outcomes for oscillations due to abiotic forcing versus oscillations due to biotic interactions.

\section{Models and Assumptions}

To understand the formation of evolutionarily stable 
coalitions of dispersive phenotypes, we consider a population consisting of 
$m$ phenotypes dispersing in an environment consisting of $n$ patches. Let 
$x_i^j(t)$ denote the abundance of phenotype $i$ in patch $j$. The fitness of an individual in patch $j$ is assumed to be of 
the form $f^j(t,\sum_{i=1}^m x_i^j(t))$. In particular, this assumption implies that phenotypes only differ demographically in their propensity to disperse. Moreover, we assume  that  $f^j(t,\cdot)$ is of period $p$. This periodicity may arise from exogenous forcing or biological interactions (e.g. over compensating density 
dependence or predator-prey interactions). While we do not explicitly model 
interactions with other species, our formulation is sufficiently general to be 
viewed as the dynamics of a particular species embedded within a web of interacting species. 

We assume that the fraction of phenotype $i$ dispersing from any given patch is $d_i$. Of the individuals dispersing from patch $j$, a fraction $S_{kj}$ go to patch $k$. We call $S=(S_{kj})$ the dispersal matrix and it characterizes how dispersing individuals are redistributed across the landscape.  Under these assumptions, the interacting phenotypes exhibit the following population dynamics: 
\begin{eqnarray*}
x_i^j(t+1)=(1-d_i)f^j\left(t,\sum_{l=1}^m x_l^j(t)\right) x_i^j(t)+
d_i \sum_{k=1}^n S_{jk} f^k\left(t,\sum_{l=1}^m x_l^k(t)\right) x_i^k(t)
\end{eqnarray*}
To express this model more succinctly, let $\bx^j=(x_1^j,\dots,x_m^j)$ be the vector of
abundances of the $m$ phenotypes in patch $j$, $\bx_i=(x_i^1,\dots,x_i^n)^T$ (where $^T$ denotes transpose) be the vector of abundances of strategy $i$ across the $n$ patches,  and $\|\bx^j\|=\sum_{i=1}^m x_i^j$ denote the total abundance of individuals  in patch $j$.  Let $\bx$ be the matrix with entries $x_i^j$, $\bF(t,\bx)$ be a diagonal matrix whose $j$-th diagonal element is $f^j(t,\|\bx^j\|)$, and 
$\bS(d_i)=(1-d_i)I+d_iS$
 where $I$ is the identity matrix. With this notation, the model is represented more succinctly as   
\begin{equation}\label{eq:model}
\bx_i(t+1)= \bS(d_i) \bF(t,\bx(t)) \bx_i(t)\qquad i=1,\dots,m.
\end{equation}

About this model, we make three assumptions throughout this manuscript. 
\begin{description}
\item[A1] The dispersal matrix $\bS$ is irreducible and column stochastic (i.e. $S$ has nonnegative entries, and the entries of each column sum up to one).
\item[A2] \eqref{eq:model} has a positive period-$p$ point 
$\hat{\bx}(1),\dots,\hat{\bx}(p)$ i.e. $\|\hat \bx_i(t)\|>0$ and $\|\hat \bx^j(t)\|>0$ for all $i$, $j$, $t$, and 
\begin{equation}\label{eq:periodic}
\prod_{t=1}^p \bS(d_i)\bF(t,\hat\bx(t)) \hat \bx_i(1)=\hat\bx_i(1)\mbox{ for all }i.
\end{equation}
\item[A3] $\bF(t,\bx)$ is continuous in $\bx$.
\end{description}
Assumption \textbf{A1} ensures that individuals or their decedents can move 
from any patch to any other patch after sufficiently many generations and there is no direct cost to dispersal. Assumption \textbf{A2} implies that the phenotypes are coexisting in a periodic fashion and occupying all the patches. Assumption \textbf{A3} is a basic regularity assumption met by most models. 

\section{Main results}

We are primarily interested in understanding when a periodically fluctuating collection of phenotypes can not be invaded by any other phenotype, and what are the spatial-temporal patterns of fitness at these potential evolutionary end states. To state our main results  precisely, we need two sets of definitions. 

\subsection*{Invasion rates, Nash equilibria and evolutionary stability.} 
Let $\bd=(d_1,\dots,d_m)$ the coalition of strategies played by the resident population. If we added a ``mutant'' phenotype with dispersal 
rate $\tilde d$ into the population and this mutant population $\by=(y^1,\dots,y^n)$ is very small, then the resident's population dynamics are initially barely influenced by the mutant's population dynamics and the mutant's dynamics are approximately given by a linear model 
\[
\by(t+1)=\bS(\tilde d) F(t,\bx(t))\by(t).
\]
By standard linearization theorems (see, e.g., \cite{katok-hasselblatt-95}), this approximation is valid when the size of the mutant population is small and the periodic point in \textbf{A2} is hyperbolic. 

The initial fate of the mutant population depends on the  \emph{invasion rate of strategy $\tilde d$ against the resident population playing strategies $\bd$}:
\[
\cI(\bd;\tilde d)=\varrho\left( \prod_{t=1}^p \bS(\tilde d)\bF(t,\hat\bx(t))
\right)  ^{1/p}
\]
where $\bx(t)$ has period $p$ and $\varrho(\mathbf{A})$ corresponds to the largest eigenvalue for a non-negative 
matrix $\mathbf{A}$. If the invasion rate $\cI(\bd;\tilde d)$ is greater than one, then the mutant population grows when its size is small. The ultimate fate of the mutant and resident after the mutant increases depends on the details of full non-linear model of the resident and invader dynamics. In particular, following an invasion the asymptotic dynamics in general may no longer be periodic (i.e. satisfy \textbf{A2}). There are  many cases where post-invasion dynamics will remain periodic (e.g. periodically forced competitive systems as discussed in section 4.1 or when there is attractor inheritance for a sufficiently small mutation~\cite{geritz-etal-02}).   

If the invasion rate $\cI(\bd;\tilde d)$ is less than one, then the mutant population declines exponentially when rare and it can not invade.  Finally, if $\cI(\bd;\tilde d)=1$, then a mutant may increase or decrease when rare depending on the details of the nonlinearities of the full model. However, if it increases, then it does so at a subexponential rate and, therefore, may be highly vulnerable to stochastic extinction. An important consequence of our assumption \textbf{A2} is that the invasion rate of mutants with the same strategy as a resident equals one i.e. 
$\cI(\bd;\tilde d)=1$ whenever $\tilde d=d_i$ for some $i$.

Using the invasion rates of mutant strategies, we can define several concepts associated with evolutionary stability. A coalition of strategies $\bd$ with $m>1$ 
is  a \emph{mixed Nash equilibrium} provided that 
\begin{equation}\label{eq:nash}
\cI(\bd;\tilde d)\le 1 \mbox{ for all } \tilde d \in [0,1]
\end{equation}
In other words, a mixed Nash equilibrium is a set of strategies in which 
all mutant strategies are unlikely to invade due to vulnerability to stochastic extinction. When $m=1$, a strategy satisfying \eqref{eq:nash} is called simply a \emph{Nash 
equilibrium}. Under the stronger assumption that rare mutants decline exponentially to extinction (i.e. $
\cI(\bd;\tilde d)<1 \mbox{ for all } \tilde d \notin\{d_1,\dots,d_m\}
$), $\bd$ is an \emph{evolutionarily stable coalition (ESC) if $m>1$} or an \emph{evolutionarily stable strategy (ESS) if  $m=1$}~\cite{cressman-03}. More generally, every ESS (respectively, ESC) is a (mixed) Nash equilibirum. 

\subsection*{Sources, sinks, and balanced patches.} 

Pulliam~\cite{pulliam-88} introduced the notion of sources and sink patches for a population at equilibrium. In source patches  birth rates exceed death rates, while in sink patches death rates exceed birth rates. Here, we extend Pulliam's definition to population exhibiting periodic fluctuations in abundance. A patch is a source if births exceed deaths ``on average'' across years, while a patch is a sink if deaths exceed births ``on average'' across years. For fitnesses varying in time, the appropriate ``average''  is the geometric mean:
\[
\bar f_j = \left(\prod_{t=1}^p f^j(t,\|\hat\bx^j(t)\|)\right)^{1/p}.
\]
If $\bar f_j<1$, then patch $j$ is a \emph{sink}.  If $\bar f_j>1$, then  patch  $j$ is a \emph{source}. Following McPeek and Holt~\cite{mcpeek-holt-92}, we say patch $j$  is \emph{balanced}  if $\bar f_j=1$. Individuals remaining in a balanced patch, on average exactly replace themselves.

\subsection*{Main Results.} 

We have two main results. Our first result implies that if there is spatial heterogeneity in the within-patch fitnesses under equilibrium conditions, then no coalition of distinctive phenotypes can coexist, there is selection for slower dispersers, and the landscape supports source and sink patches. In particular, this result implies that for populations at equilibrium and playing a Nash equilibrium, all patches must be balanced. This result follows from \cite{siap-06} and generalizes earlier work of Hasings~\cite{hastings-83} and Parvinen~\cite{parvinen-99} by allowing for non-diffusive patterns of dispersal (i.e. $S$ need not be symmetric). 

\begin{proposition}\label{prop:one} If $p=1$ and $\bF(1,\hat \bx(1))$ is not 
scalar (i.e. there is spatial heterogeneity), then all the $d_i$ are equal and positive, $\cI(\bd;\tilde d)>1$ for all $0\le  \tilde d< d_1$, and  at least one patch is a sink and at least one patch is a source. 
\end{proposition}

\begin{proof}
Let $\bA$ be the diagonal matrix $\bF(1,\hat\bx(1))$ with diagonal entries $A_{jj}=f^j(1,\|\hat \bx^j(1)\|)$.
Since $\bA$ is not scalar, Theorem~3.1 in \cite{siap-06} implies that  $\varrho(\bS(d)\bA)$ is a strictly decreasing function of $d\in [0,1]$. Assumption \textbf{A2} implies that $\varrho(\bS(d_i)\bA)=1$ for all $i$. Hence, all of the $d_i$ are equal and $\cI(\bd;\tilde d)>1$ for all $0\le \tilde d <d_1$. The $d_i$ can not equal zero as this would imply that $\bF(1,\hat \bx_i(1))$ is the identity matrix contradicting the assumption that it is non-scalar. Since $\bS(d_1)$ is column stochastic and $\bA$ is non-scalar, $\max_i A_{ii}>\varrho(\bS(d_1)\bA)=1>\min_i A_{ii}$. Hence, there is a source patch and a sink patch. 

\end{proof}

Our second result concerns period $2$ environments. In contrast to populations at equilibrium, this result implies  that any ESS or ESC (or more generally, Nash equilibrium) includes a dispersive phenotype, supports at least one sink patch and supports no source patches. We are able to prove this result only under the restriction that the dispersal matrix $S$ is \emph{diagonally similar to a symmetric matrix}: there exists an invertible diagonal matrix $D$ such that $DSD^{-1}$ is a symmetric matrix. This allows for 
a diversity of movement patterns including diffusive movement (i.e. symmetric 
$S$) and any form of local movement along a one-dimensional gradient (i.e. tridiagonal $S$). A proof is given in Appendix A. It is an open problem whether this result extends to all irreducible stochastic matrices $S$.

\begin{theorem}\label{thm:main} If $p=2$, $S$ is diagonally similar to a 
symmetric matrix, $\bF(2,\hat\bx(2))\neq \bF(1,\hat\bx(1))$ (i.e. there is temporal heterogeneity), and 
$\bF(t,\hat\bx(t))$ is non-scalar for some $t$ (i.e. there is some spatial heterogeneity), then 
\begin{enumerate}
\item[(i)] If $\max_id_i=0$,  then $\cI(\bd;\tilde d)>1$ for all $0<\tilde d\le 1$. 
\item[(ii)] If $\bd$ is a (possibly mixed) Nash equilibrium, then $\max_i d_i>0$, the set 
$\cS\subset\{1,\dots,n\}$ of sink patches is non-empty, and the remaining 
patches $\{1,\dots,n\}\setminus \cS$ are balanced.
\end{enumerate}
\end{theorem}

\section{Applications}

To illustrate the utility of our results, we present two applications. The first  application considers populations with compensating density dependence in a periodically forced environment. We use Theorem~\ref{thm:main} and its proof to determine under what conditions there is an ESC of sedentary and dispersive phenotypes. The second application considers populations with sufficiently overcompensatory density dependence to create oscillatory dynamics. We use the proof of Theorem~\ref{thm:main} to illustrate how the evolution of higher dispersal rates can synchronize initially asynchronous population dynamics.

\subsection{Evolution of dispersal dimorphisms in periodic environments}
We consider competing dispersive phenotypes whose within patch dynamics are given by a periodically forced Beverton-Holt model. For simplicity, we assume that only the 
intrinsic fitness $\lambda_t^j$ of an individual living in patch $j$ varies in 
time and space:
\begin{equation}
f^j(t,\|\bx^j\|)=\frac{\lambda^j_t}{1+a \|\bx^j\|}
\end{equation}
where $a$ measures the intensity of competition within a patch. Also for simplicity, we assume that 
$S_{kj}=1/n$ for all $j,k$. In other words, dispersing individuals are 
uniformly distributed across the landscape.  

When there is only one dispersive phenotype (i.e.  $m=1$),  equation \eqref{eq:model} is a sublinear monotone map (see, e.g., \cite{hirsch-smith-05} for definitions).  Consequently, \cite[Thm. 6.1]{hirsch-smith-05} implies that the fate of the population depends on the linearization of the system at the extinction state. The dominant eigenvalue associated with this linearization is given by \[
\lambda=\varrho\left(\prod_{t=1}^p\bS(d_1)\bF(t,0)\right)^{1/p}.\]
If $\lambda\le 1$, then population goes deterministically toward extinction i.e.  $\bx(t)$ converges to $0$ for all $\bx(0)\ge 0$.  Alternatively, if $\lambda>1$, then the populations increases when rare and ultimately converges to a periodic orbit. More precisely, there exists a periodic orbit, 
$\{\hat \bx(1),\dots,\hat \bx(p)\}$ with $\hat \bx(t)\gg0$ 
for all $t$, such that $\bx(t)$ converges to this periodic orbit whenever 
$\bx(0)\gg 0$.  A sufficient condition ensuring $\lambda>1$  is 
\begin{equation}\label{eq:suff}
\left(\prod_{t=1}^p \lambda^j_t\right)^{1/p} >1 \mbox{ for all }j.
\end{equation}
 In other words, all the patches can support a population  in the absence of immigration. While this condition is stronger than what is necessary, for simplicity, we 
assume that \eqref{eq:suff} holds for the remainder of this section.

\begin{figure}
\begin{center}
\includegraphics[width=5in]{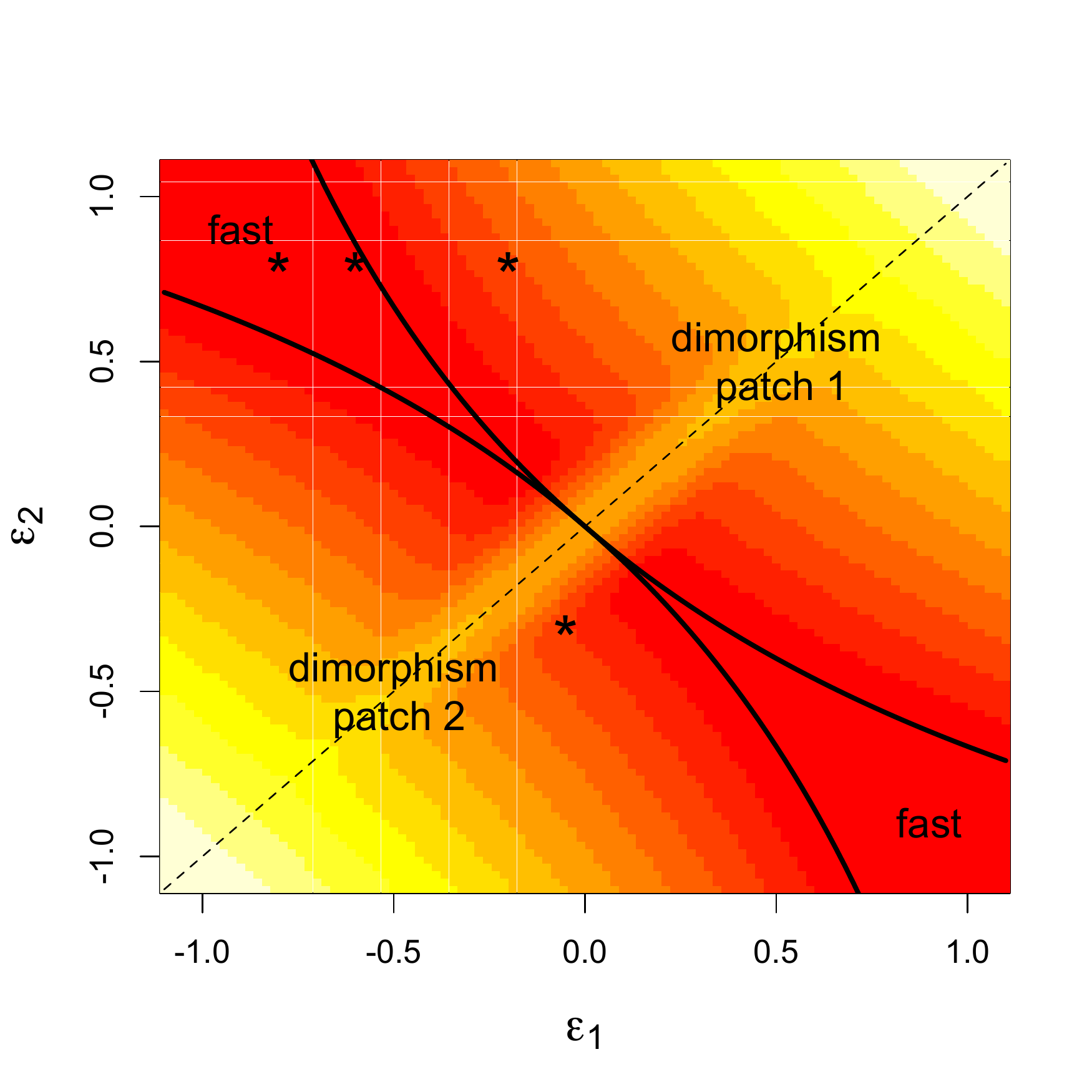}
\end{center}
\caption{Nash equilibria for a Beverton-Holt two patch model with $\lambda_t^1=2+\epsilon_t$ 
and $\lambda_t^2=2-\epsilon_t$. The strategy $d_1=1$ is a local ESS in the regions 
denoted ``fast''. The mixed strategy $\bd=(1,0)$ is an ESC in the regions 
denoted ``dimorphism''. For these regions, the sedentary population only 
resides in the indicated patch. Shading corresponds to the average abundance 
of the sedentary strategy with lighter colors corresponding to higher 
abundance. At the dashed line, all (coalitions of) strategies are a (mixed) Nash 
equilibrium. $*$s refer to parameter values for which pairwise invasibility plots are shown in Fig.~\ref{fig:pip}}\label{fig:2patch} 
\end{figure}

For period two environments where the intrinsic fitness vary in space and 
time, Theorem~\ref{thm:main} implies that  a sedentary strategy (i.e. $d_1=0$) is not a Nash equilibrium as 
it can be invaded by more dispersive phenotypes.  In 
contrast, for a fully dispersive phenotype (i.e. $d_1=1$), there is a periodic point $\hat \bx(t)$ given by 
\[
\hat x^j_1(t)=\frac{\bar\lambda_2\bar\lambda_1-1}{a(1+\bar \lambda_t)}
\]
where $\bar \lambda_t=\frac{1}{n} \sum_{j=1}^n \lambda^j_t$ is the spatial average of the intrinsic fitnesses. Along this 
periodic orbit, a computation reveals that the within-patch fitnesses satisfy 
\[
\prod_{t=1}^2 f^j(t,\|\hat \bx^j(t)\|)= \prod_{t=1}^2 \frac{\lambda^j_t}{\bar 
\lambda_t}. 
\]
Thus, Theorem~\ref{thm:main} implies that a necessary condition for $d_1=1$ 
to be a Nash equilibrium is 
\begin{equation}\label{eq:fast}
\prod_{t=1}^2 \lambda^j_t \le \prod_{t=1}^2 \bar\lambda_t \mbox{ for all }j
\end{equation}
with a strict inequality for at least one $j$. This condition for a Nash equilibrium requires that the geometric mean of the fitness within each patch is no greater than geometric mean of the spatially averaged fitness. 

To illustrate the utility of \eqref{eq:fast}, consider an environment where the fitness in each patch fluctuates between a low value $\lambda_{bad}$ in ``bad'' years and a higher value $\lambda_{good}$ in ``good" years i.e. $\lambda_t^j \in 
\{\lambda_{good},\lambda_{bad}\}$, $\lambda_{bad}<\lambda_{good}$, and $\lambda^j_t\neq \lambda^j_{t+1}$. To ensure that all dispersal phenotypes can persist, we assume that the geometric mean $\sqrt{\lambda_{good}\lambda_{bad}}$  is greater than one.  Provided that there is some spatial asynchrony (i.e. 
$\lambda_t^j \neq \lambda_t^k$ for some $k,j$), the necessary condition \eqref{eq:fast} for a Nash equilibrium of highly dispersive phenotypes simplifies to 
\[
\sqrt{\lambda_{good}\lambda_{bad}}\le  \frac{1}{2}(\lambda_{good}+\lambda_{bad}).
\]
This inequality holds strictly as the geometric mean is less than the arithmetic 
mean. Furthermore, a computation reveals that the geometric mean of fitness within patch $j$ satisfies
\[
\sqrt{\prod_{t=1}^2 f^j(t,\|\hat \bx^j(t)\|)}=\sqrt{ 
\frac{\lambda_{good}\lambda_{bad}}
{\frac{1}{4}(\lambda_{good}+\lambda_{bad})^2}}<1 
\]
for all patches $j$. Hence, Theorem~\ref{thm:1} in Appendix A implies that $\cI(1;\tilde d)<1$ for all $\tilde d \in [0,1)$. Therefore, for this environment, a highly dispersive phenotype ($d_1=1$) always is an ESS  and  all patches are sinks for populations playing this ESS. 

 \begin{figure}
\begin{center}
\begin{tabular}{cc}
\includegraphics[width=2.5in]{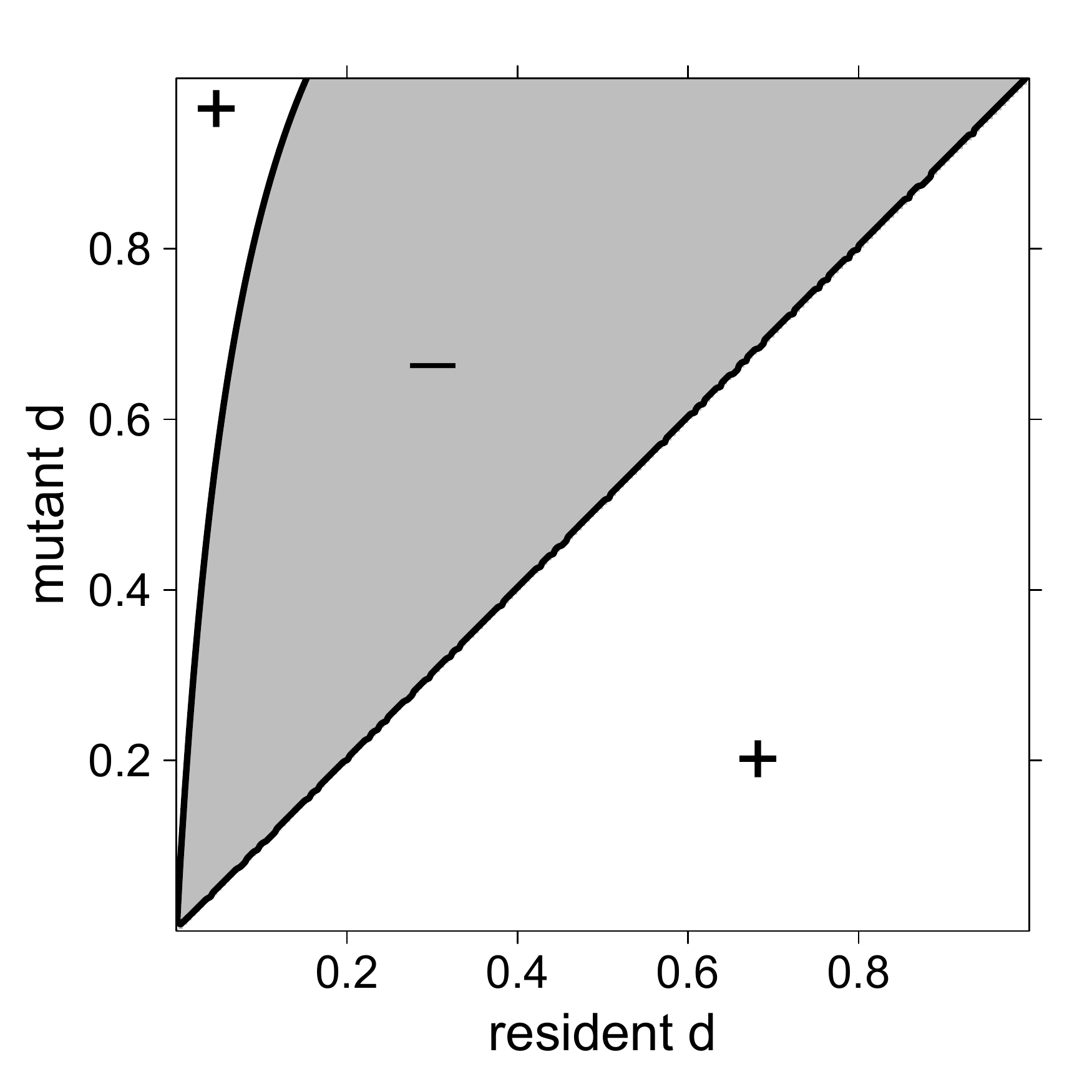}&\includegraphics[width=2.5in]{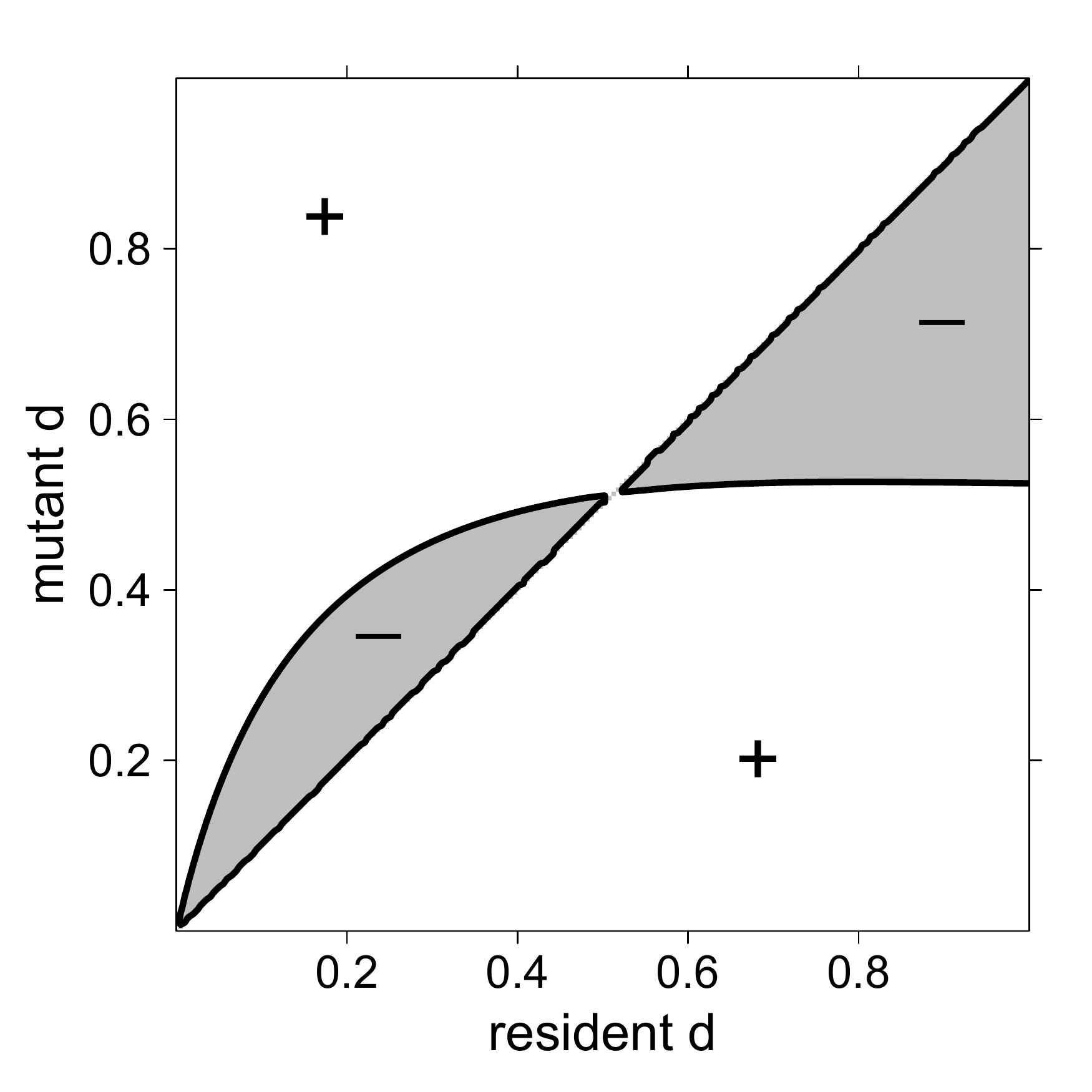}\\
(a) $\epsilon_t=-0.05,-0.3$& (b) $\epsilon_t=-0.2,0.8$\\
\includegraphics[width=2.5in]{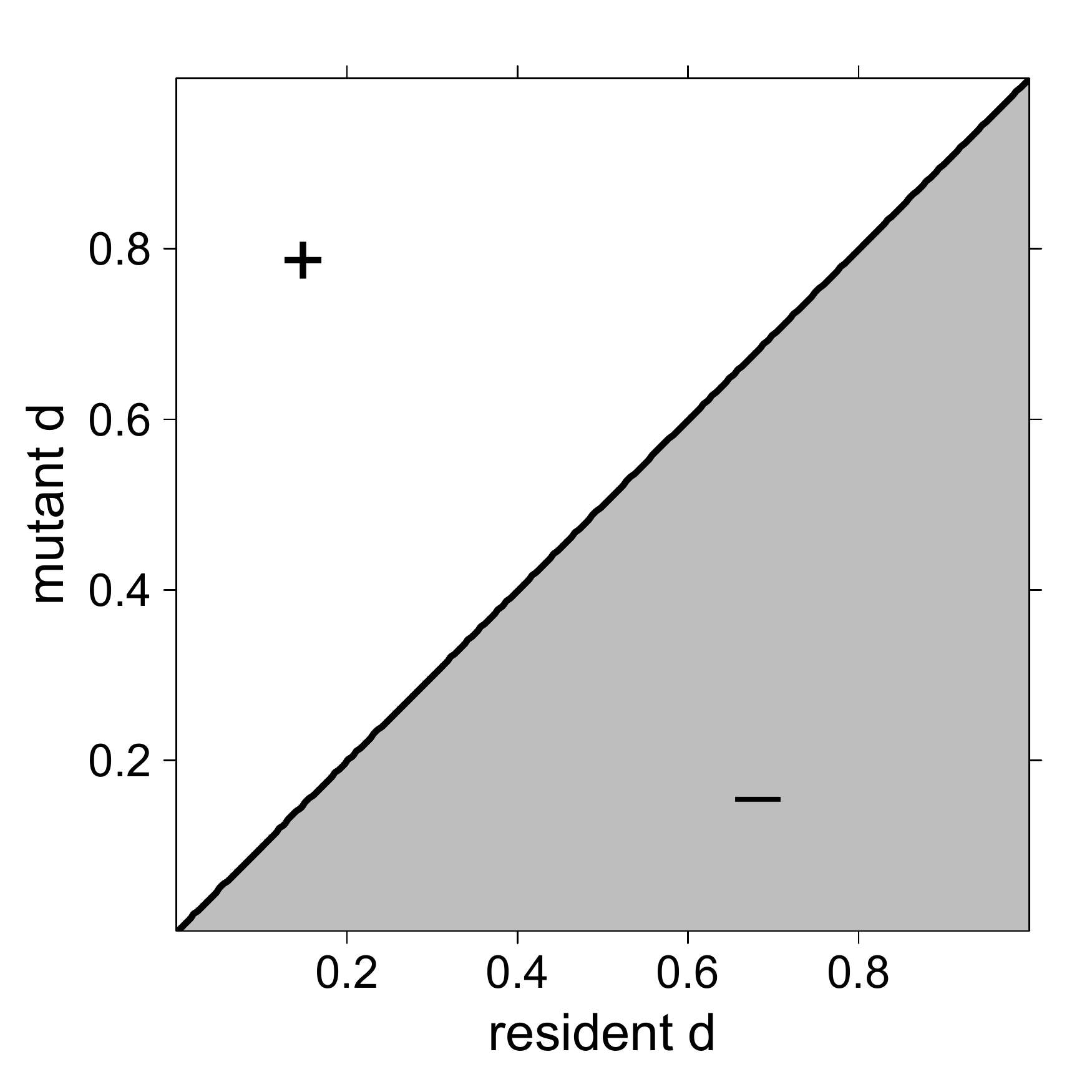}&\includegraphics[width=2.5in]{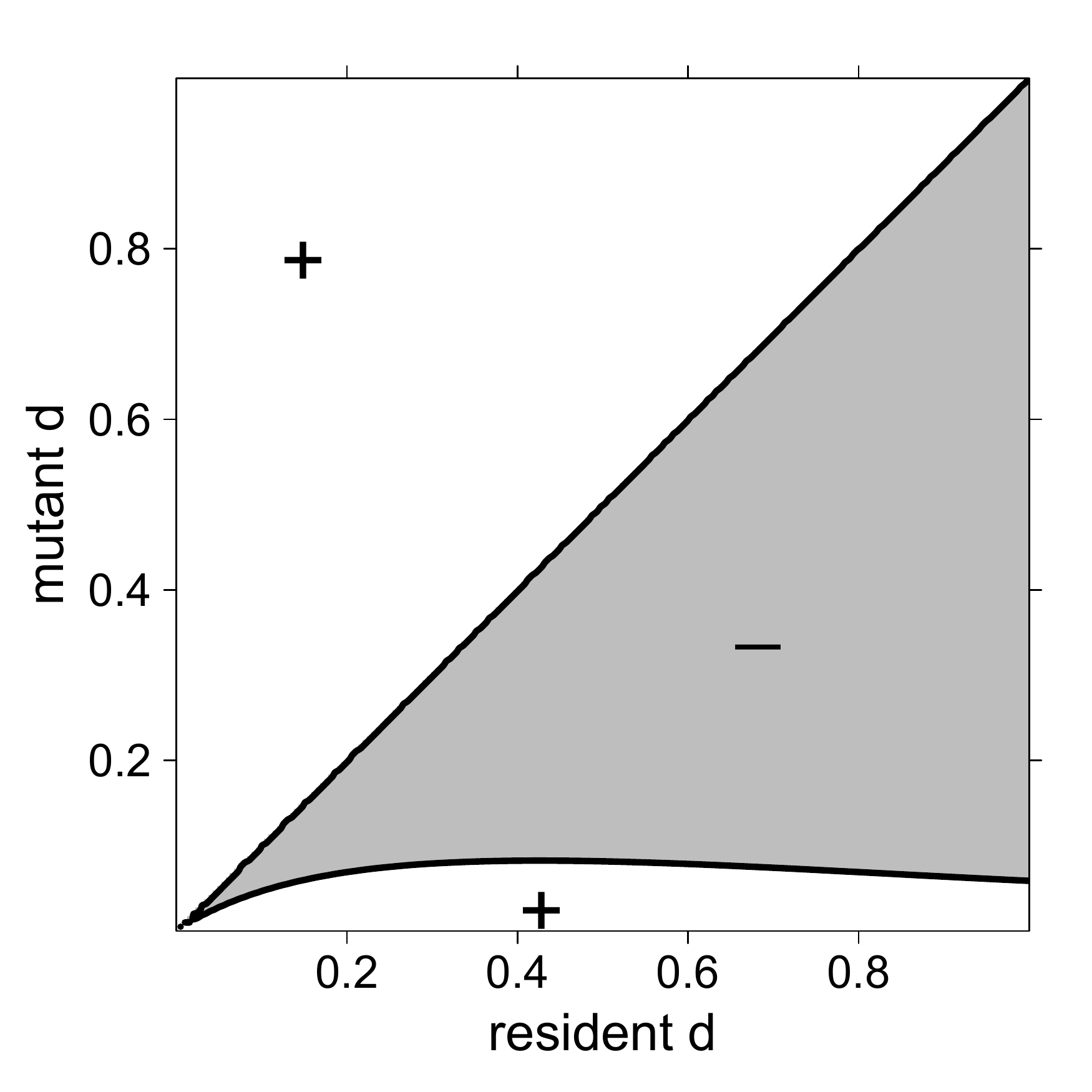}\\
(c) $\epsilon_t=-0.8,0.8$& (d) $\epsilon_t=-0.6,0.8$
\end{tabular}
\end{center}
\caption{Pairwise invasibility plots (PIPs) for a two patch Beverton-Holt model with $\lambda_t^1=2+\epsilon_t$  and $\lambda_t^2=2-\epsilon_t$. The horizontal and vertical axes correspond to resident $d_1$ and mutant $\tilde d$ dispersal rates, respectively. The dark lines, shaded regions, and unshaded regions correspond to where $\cI(d_1,\widetilde d)=1$, $\cI(d_1,\widetilde d)<1$, and $\cI(d_1,\widetilde d)>1$, respectively. In (a) and (b), $d_1=0$ is an evolutionary branching point. In (b), $d_1\approx 0.5$ is a convergently unstable singular point. In (c), $d_1=1$ is convergently stable and evolutionarily stable. In (d), $d_1=1$ is convergently stable and a local ESS, but invasable by sufficiently sedentary phenotypes.  }\label{fig:pip}
\end{figure}

When the necessary condition \eqref{eq:fast} for a highly dispersive phenotype to be a Nash equilibrium is not met, sedentary dispersers can invade the patches 
whose average fitness $\sqrt{\prod_{t=1}^2 \lambda^j_t}$ exceeds the spatially averaged fitness $\sqrt{\prod_{t=1}^2 \bar\lambda_t}$. To understand the implications of this invasion, we consider a two patch environment where environmental fluctuations are spatially asynchronous. More precisely, we assume that spatial average of fitness does not vary in time (i.e. $\bar \lambda_t=\bar \lambda$)  and the temporal fluctuations $\epsilon_t\in (-\bar \lambda,\bar \lambda)$ around this spatial average are asynchronous i.e.  
$\lambda^1_t=\bar\lambda+ \epsilon_t$ and $\lambda^2_t =\bar\lambda-
\epsilon_t$ for $t=1,2$. The necessary condition \eqref{eq:fast} for a highly dispersive phenotype to be a Nash equilibrium becomes 
\begin{equation}\label{nec}
\bar\lambda |\epsilon_1+\epsilon_2|\le  |\epsilon_1 \epsilon_2| \mbox{ and } 
\epsilon_1\epsilon_2<0 .
\end{equation}
When \eqref{nec} is not meet, extensive numerical simulations suggest that after the successful invasion of the  sedentary dispersers, the populations approach a period $2$ orbit $\hat  \bx(t)$. Moreover, these simulations 
suggest that $\bd=(1,0)$ is an ESC and, consequently, a potential evolutionary end point consisting of a dimorphism of sedentary and highly dispersive phenotypes. At this ESC, one patch is balanced and occupied by both phenotypes, while the other patch is  a sink and only occupied by the dispersive phenotype. \eqref{nec} implies that the ESC occurs when the temporal correlations of within patch fitness are not too negative (i.e. $\epsilon_1$ is not too close to $-\epsilon_2$ in Fig.~\ref{fig:2patch}). Pairwise invasibility plots (see, e.g., \cite{geritz-etal-97} for a discussion of the interpretation of these plots and the associated terminology) suggest that these ESCs can be reached by small mutational steps when $d=0$ is a convergently stable branching point  (Fig.~\ref{fig:pip}a,b). On the other hand, when \eqref{nec} is satisfied, the highly dispersive phenotype may or may not be an ESS in the strict sense (Fig.~\ref{fig:pip}c,d). Although numerical simulations confirm that the highly dispersive phenotyperesists invasion attempts by nearby phenotypes (i.e. $\cI(1,\tilde d)\le 1$ whenever $\tilde d$ is sufficiently close to $1$), relatively sedentary phenotypes still may be able invade (Fig.~\ref{fig:pip}d).

\subsection{Evolution of synchronicity}

Biotic interactions can generate oscillatory dynamics and, thereby, temporal 
variation in fitness. To illustrate the feedbacks between evolution of 
dispersal and biotically generated oscillations, we consider an extension of the coupled 
Logistic map introduced by Hastings~\cite{hastings-93}. In this model, the local dynamics are determined by the Logistic fitness function $rx(1-x)$. A fraction $d$ of all individuals disperses randomly to all patches i.e. a fraction $d/k$ of individuals from patch $j$ arrives in all patches. Under these assumptions, the dynamics of a single dispersive phenotype is given by
 \[
x^j_{t+1}=(1-d) r\,x_t^j(1-x_t^j) + \frac{d}{n} \sum_{k=1}^n r\,x_t^k(1-x_t^k)
\]
We note that in the case of $n=2$ patches, Hasting's $d$ corresponds to our $d/2$. 

\begin{figure}
\begin{center}
\begin{tabular}{cc}
\includegraphics[width=3in]{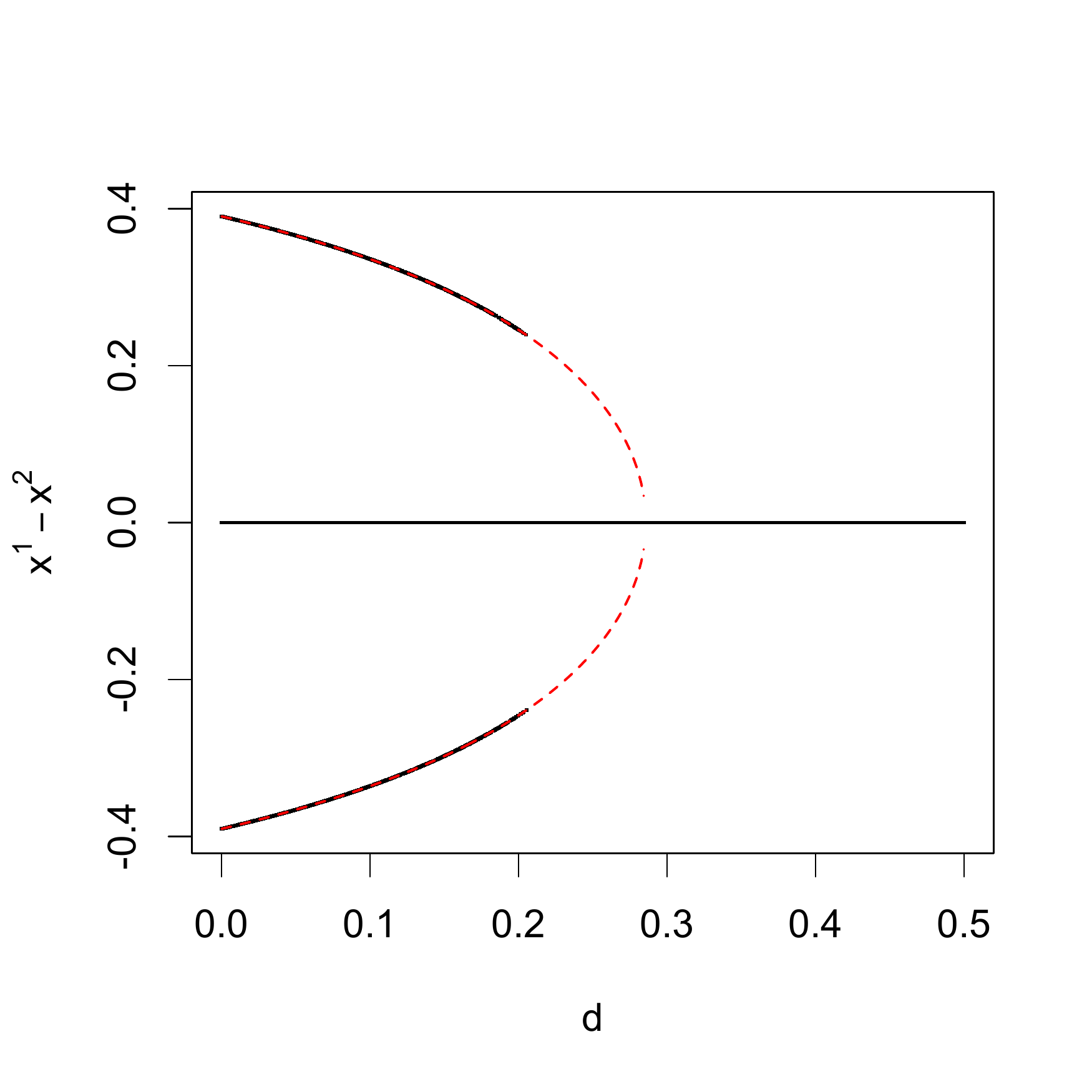}
&\includegraphics[width=3in]{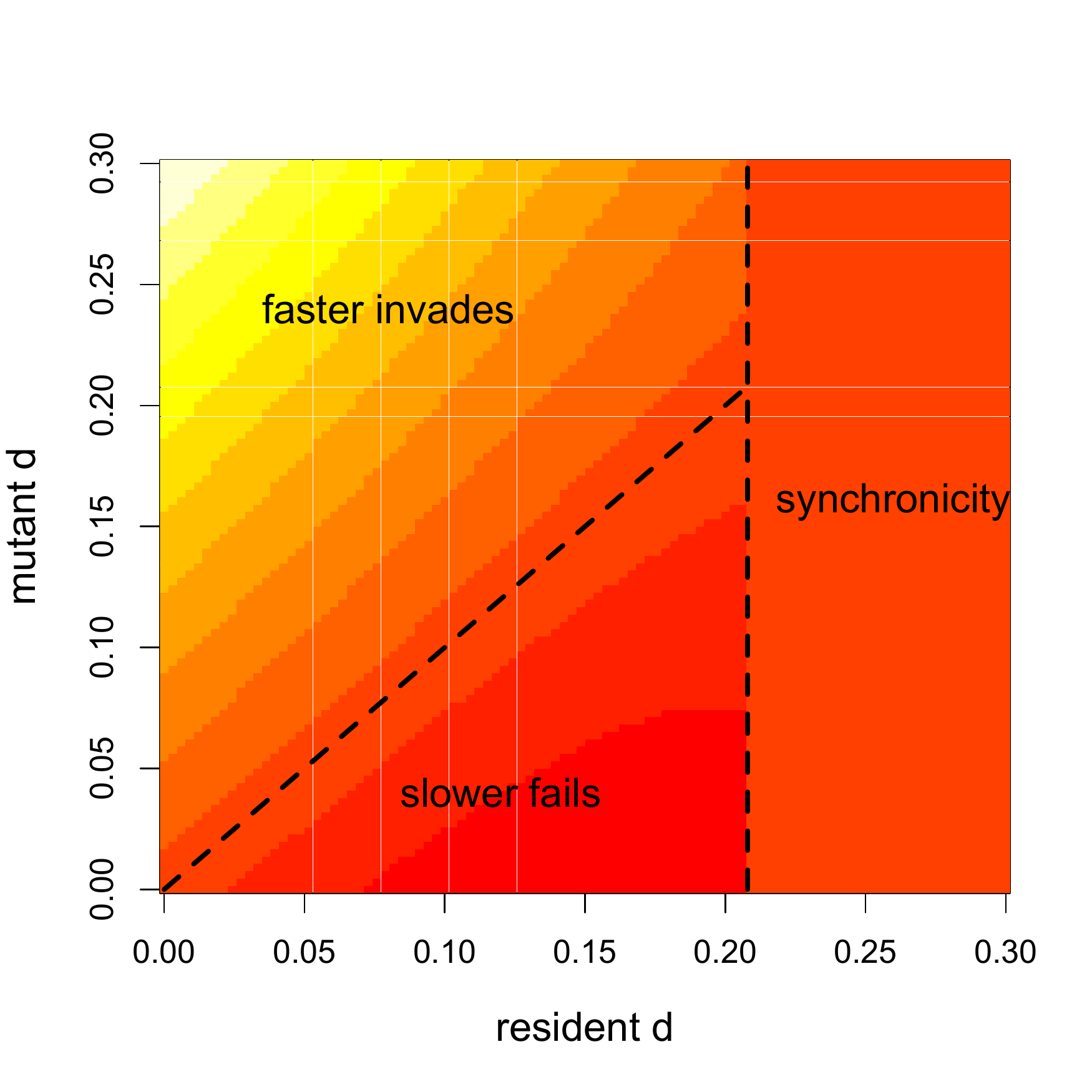}\\
(a) $r=3.4$&(b) $r=3.4$\\
\includegraphics[width=3in]{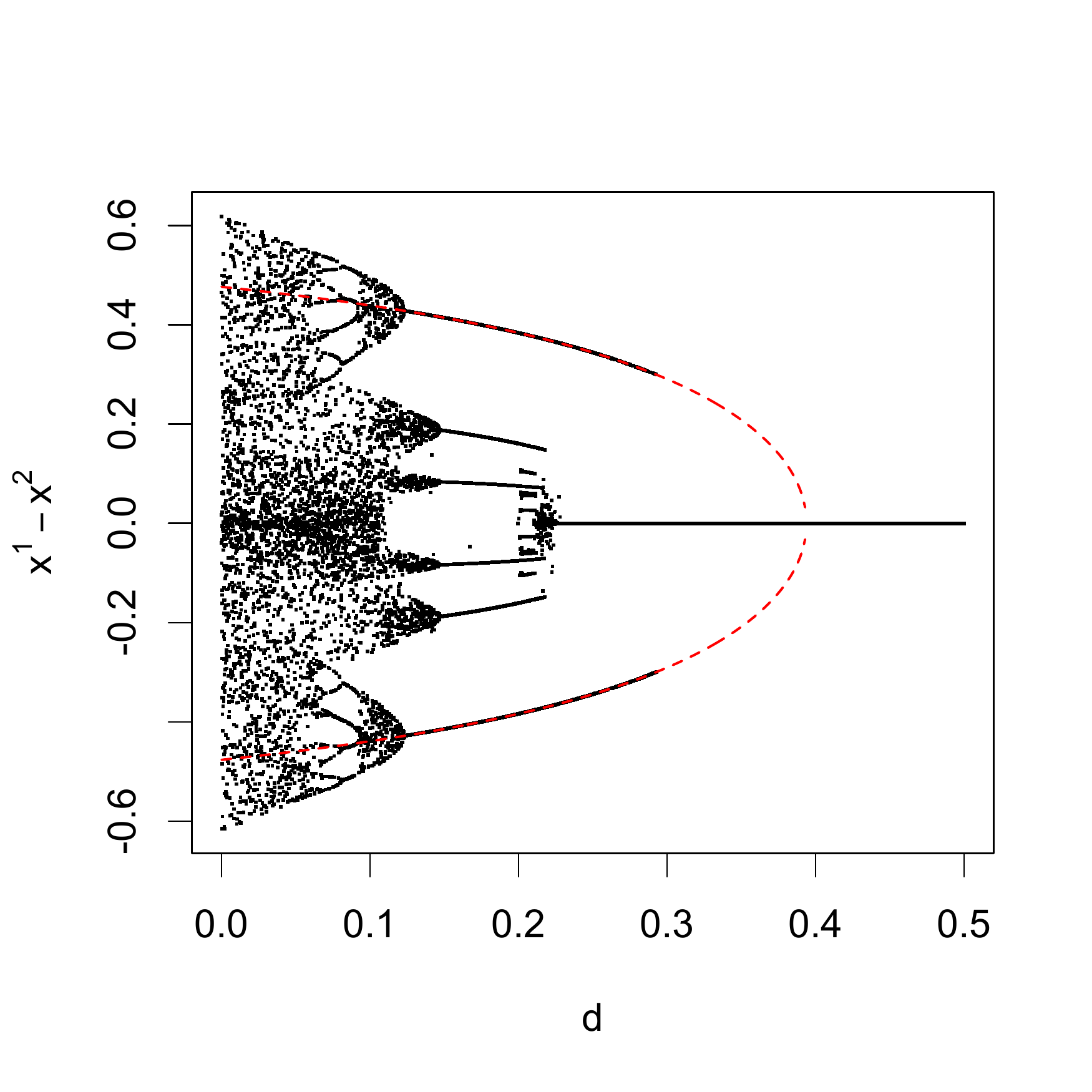}
&\includegraphics[width=3in]{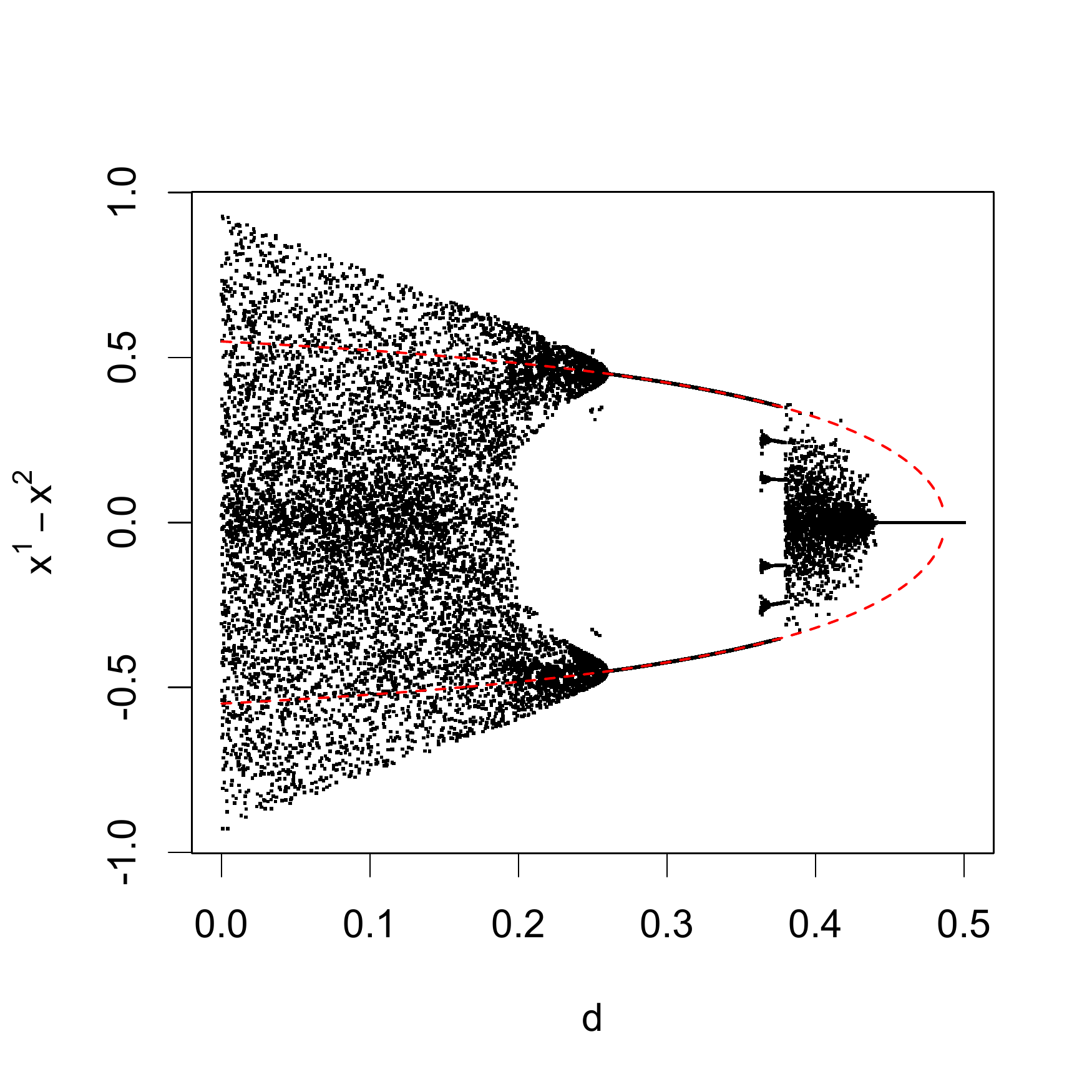}\\
(c) $r=3.65$ &(d) $r=3.95$ \end{tabular}
\caption{Evolution of spatial synchronization in a two-patch Logistic model. In (a), (c), (d), orbital bifurcation diagrams for phase difference $x^1-x^2$ are shown for for the two-patch Logistic equation. All phase differences along attractors are plotted in black. The red dashed curve corresponds to the out-of-phase, period two orbit which is stable only when it overlaps the black regions. In (b),  the contours of $\cI(d;\tilde d)$ are plotted along the out-of-phase periodic 
point whenever it is stable. }\label{fig:logistic}
\end{center}
\end{figure}

When  $r>3$ and there are two patches, Hastings~\cite{hastings-93} has shown that there is an interval of dispersal rates between $0$ and $1$ such that there is an out-of-phase stable period two  point (Figs.~\ref{fig:logistic}a,c,d). To apply our results, let $m=1$, $d_1=d$ for which there is a stable out-of-phase periodic point, $\hat \bx(1),\hat \bx(2)$  (an explicit formula for this orbit can be found in the Appendix of  \cite{hastings-93}), $S_{jk}=1/2$ for $1\le j,k\le 2$, and $f^j(t,\|\bx^j\|)=r(1-\|\bx^j\|)$. Along this out-of-phase periodic orbit, $\prod_t \bF(t,\hat\bx(t))$ is a scalar matrix. Hence, Theorem.~\ref{thm:1} implies that along this out-of-phase periodic orbit $\cI(d_1,\tilde d)>1$ for all $d_1<\tilde d \le 1$. Hence, there would be selection for higher dispersal rates. For this two patch case, this implication of Theorem~\ref{thm:1} also follows from a proposition of Doebeli and Ruxton~\cite[pg. 1740]{doebeli-ruxton-97} in which they performed a direct calculation of the eigenvalues.  For $3<r\le  3.7$, numerical simulations suggest that this selection for higher dispersal rates ultimately results in a dispersal rate that spatially synchronizes the dynamics (Figs.~\ref{fig:logistic}a-c) at which point all dispersal rates are Nash equilibria. However, at higher $r$ values such as $r=3.95$ (Fig.~\ref{fig:logistic}d), the destabilization of the out-of-phase period $2$ point (at $d_1\approx 0.38$) results in more complex asynchronous dynamics in which case our results are not applicable and evolution of dispersal may no longer synchronize the dynamics. 

\begin{figure}
\begin{center}
\begin{tabular}{cc}
\includegraphics[width=3in]{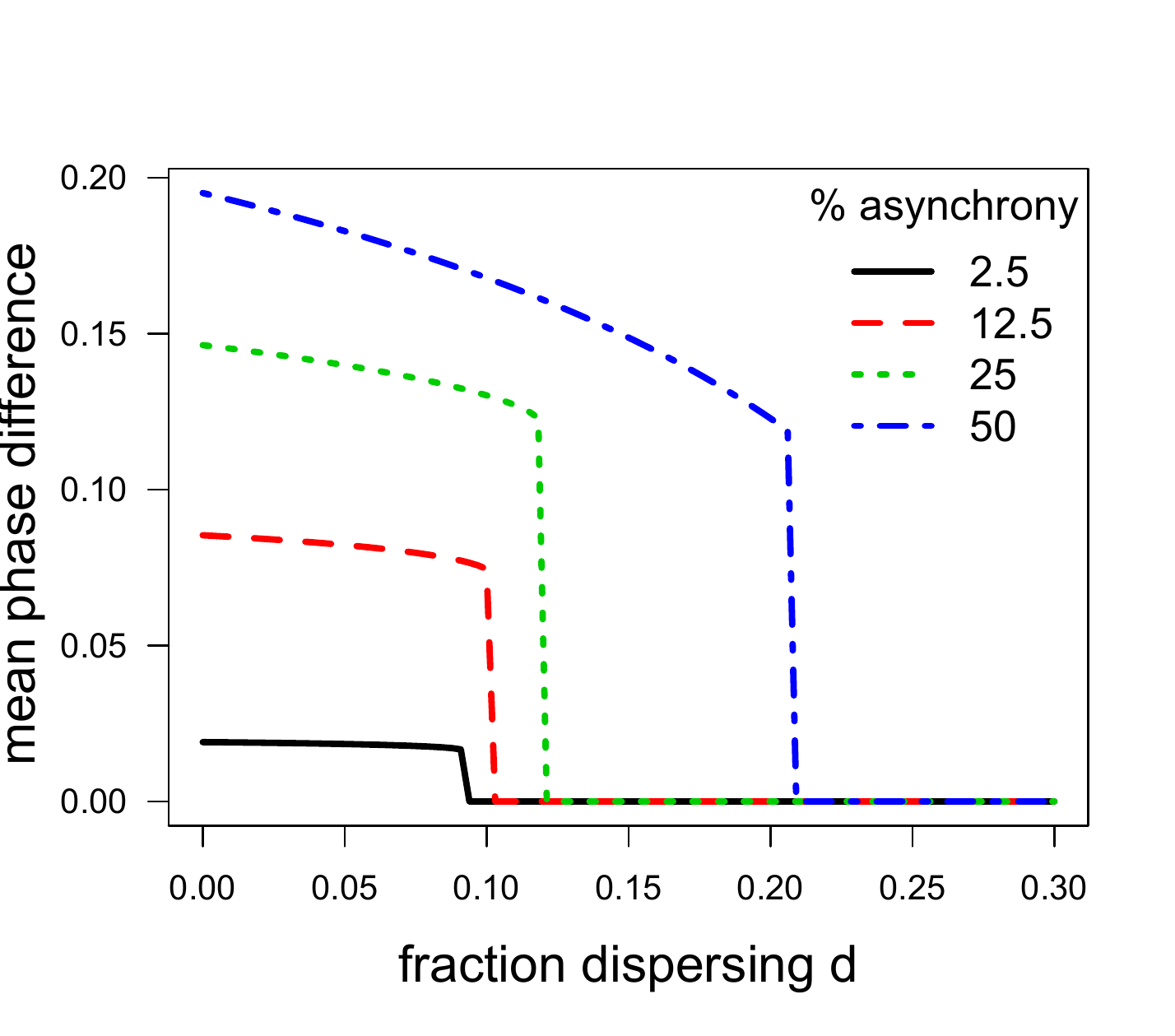}
&\includegraphics[width=3in]{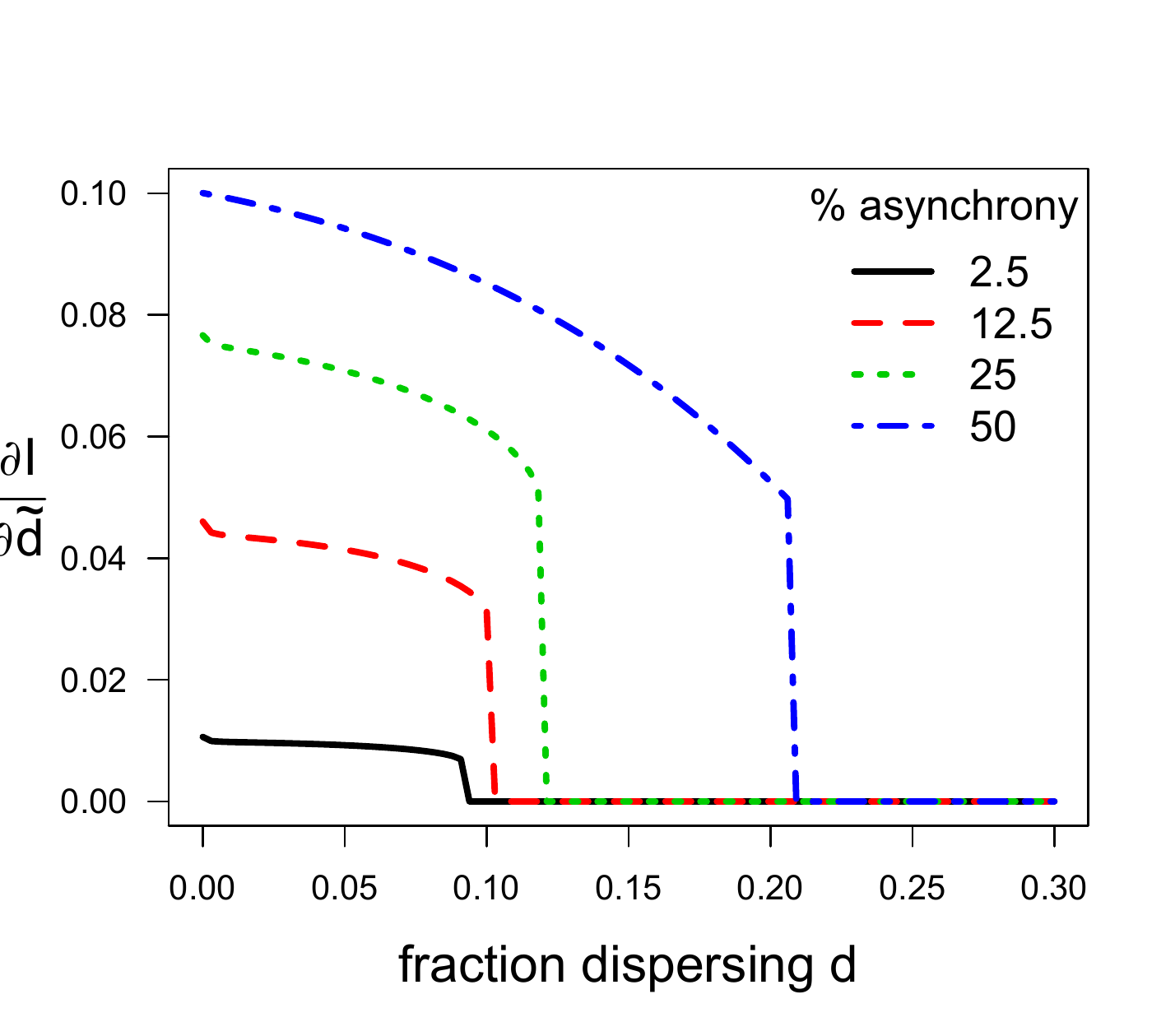}\\
(a) &(b)  \end{tabular}
\caption{Evolution of spatial synchronization for a 40 patch Logistic model. 
In (a), the mean phase difference in spatial abundance (i.e. $\frac{1}{40^2}\sum_{j\neq k} |x^j-x^k|$) for a single population exhibiting a two-cycle is plotted as function of its dispersal rate. In (b), the strength of selection $\frac{\partial \cI}{\partial \tilde d}(d,d)$ for higher dispersal rates is plotted as a function of its dispersal rate. Different lines  correspond to different percentages of initial asynchrony for sedentary phenotype e.g. 50\% implies that half the patches at $d=0$ where out of sync with the other set of patches.}  \label{fig:logistic2}
\end{center}
\end{figure}

When there are more than two patches and $r>3$, out-of-phase $2$ cycles can take on a greater diversity of forms. In particular, one can divide the landscape into two sets of patches such that patches are synchronous within each set and asynchronous across sets. Due to this potential spatial asymmetry in these out-of-phase cycles, we have not been able to show the geometric mean of fitness equals one in all patches. However, numerical simulations for $n=40$ patches and $3<r<3.45$, suggest that this does occur. In which case, Theorem~\ref{thm:1} implies that there would selection for higher dispersal rates along these asynchronous cycles. In fact, numerical simulations for $3<r<3.45$ show that there would be selection for higher dispersal rates until the dynamics are spatially synchronized (Fig.~\ref{fig:logistic2}). Moreover, these simulations show, quite intuitively, greater initial asynchrony for the dynamics of the sedentary phenotype result in a stronger selection gradient (Fig.~\ref{fig:logistic2}b) and require the evolution of higher dispersal rates to regionally synchronize the dynamics (Fig.~\ref{fig:logistic2}a).  

\section{Discussion}

We analyzed the evolution of dispersal in spatially and temporally variable environments. When there is spatial variation in fitness and within patch fitness varies in time between a lower and higher value, we proved that any evolutionary stable strategy (ESS) or evolutionarily stable coalition (ESC) includes a dispersive phenotype. In particular, our results imply a sedentary population can be always be invaded by more dispersive phenotypes. These results are particularly remarkable in light of earlier analysis showing generically the only ESS for spatially heterogenous environments without temporal heterogeneity is a sedentary phenotype~\cite{hastings-83,dockery-etal-98,parvinen-99,siap-06}. Hence, our results imply this earlier work on discrete-time models is not generic: arbitrarily small periodic perturbations result in selection for dispersal.  

Our results extend a result of \citet{parvinen-99} who proved sedentary phenotypes could be invaded by phenotypes that disperse uniformly to all patches in every generation. Moreover, they are consistent with the numerical work of Mathias et al.~\cite{mathias-etal-01} who found that ESCs always included highly dispersive phenotypes in discrete-time models with random variation in the vital rates. In contrast, our results are only partially consistent with the analysis of reaction-diffusion models by Hutson et al.~\cite{hutson-etal-03}. Hutson et al. proved that certain forms of spatio-temporal heterogeneities select for the higher dispersal rate. However, if the frequency of spatial oscillations (i.e. spatial variation) is too large or too small, then phenotypes with higher dispersal rates are driven to extinction. A partial explanation for this discrepancy is the continuity of the per-capita growth rates in the reaction-diffusion models results in positive correlations in time which may select for slower dispersal rates. 

Our analysis and numerical simulations suggest that there are two evolutionary end states for environments where spatial-temporal heterogeneity is generated by abiotic periodic forcing: either an ESS consisting of a highly dispersive phenotype or an ESC consisting of a highly dispersive phenotype and a sedentary phenotype. This is partially consistent with the numerical work of Mathias et al.~\cite{mathias-etal-01} who always found ESCs of high and low dispersal phenotypes. Mathias et al. found that these ESC always could be achieved by small mutational steps leading to an intermediate phenotype (a convergent singular strategy) at which coalitions of faster and slower dispersers could invade and coexist (evolutionary branching). Our numerical and analytic results  differ from these conclusions in two ways. First, evolutionary branching occurs at sedentary phenotypes not intermediate phenotypes. Consequently, prior to the branching event, one of the strategies in the ESC is already present. Second, while evolution by small mutational states may culminate in a highly dispersive phenotype, an ESC of highly dispersive and sedentary phenotypes may arise by the invasion of sufficiently slow dispersers i.e. large mutational steps are required to realize the ESC. 

When spatial and temporal heterogeneity is created purely by biotic interactions and initial conditions (e.g. the coupled Logistic map), we have shown there can be selection for higher dispersal rates that ultimately may synchronize the population dynamics across space. Following this synchronization event, dispersal would become selectively neutral. Our analytic results extend Doebeli and Ruxton's~\cite{doebeli-ruxton-97} two-patch analysis to multiple patches. In a multispecies context, Dercole et al.~\cite{dercole-etal-07} found a related result. Numerical simulations of tritrophic communities with chaotic dynamics showed that  evolution of dispersal often drove these spatial networks to weak forms of synchronization. 

Our analysis states that an ESS or a ESC for dispersal result in some patches being sinks  (i.e. geometric mean of fitness less than one) and the remaining patches being balanced (i.e. geometric mean of fitness equal to one). In the case of ESCs of highly dispersive and sedentary phenotypes, the sedentary phenotypes only reside in the balanced patches. In the case of ESS of highly dispersive phenotypes, all patches may be sinks despite the population persisting. In contrast, evolutionarily stable strategies for conditional dispersal in purely spatially heterogenous environments result in all occupied patches being balanced~\cite{siap-06,cantrell-etal-07} and, thereby, exhibiting an ideal and free distribution~\cite{fretwell-lucas-70}. By appropriately modifying empirical methods  for distinguishing between source-sink dynamics and balanced-dispersal~\cite{doncaster-etal-97,diffendorfer-98}, one might be able to find empirical support for these alternative, evolutionarily stable, spatial-temporal patterns of fitness. 

Despite extensive progress, there are many mathematical challenges to overcome if we want to have a better analytical understanding of the evolution of dispersal. Two challenges, of particular interest here, are going beyond the assumption period 2 environments, and accounting for various forms of costs for dispersal. While period $2$ environments can be viewed as a crude cartoon of seasonal environments, most environmental fluctuations exhibit multiple modes in their Fourier decomposition and have a significant stochastic component~\cite{vasseur-yodzis-04}. Despite some progress~\cite{wiener-tuljapurkar-94,prsb-10}, a detailed analytic understanding of the interactive effects of this environmental stochasticity and dispersal on regional fitness remains largely elusive. Alternatively, the ability to disperse or the act of dispersing  often is accompanied by costs to the individual. Dispersing individuals may die before reaching their destination. Alternatively, there may trade-offs between dispersal ability and competitive abilities~\cite{levins-culver-71,yu-wilson-01}. If these costs or tradeoffs are sufficiently strong, they can substantially alter the predictions presented here. For example, \citet{parvinen-99} has shown, quite intuitively, that if costs to dispersal are sufficiently strong, there is always selection against dispersal. However, we need  the development of new mathematical methods to unravel the implications of intermediate costs on the evolution of dispersal in environments with spatio-temporal variation.

\appendix

\section{Proof of Theorem~\ref{thm:main}}

The proof of Theorem~\ref{thm:main} depends on the following key result which 
is proven in Appendix B.  We do not impose the assumption that $S$ is irreducible
in this result.

\begin{theorem} \label{thm:1}
Let $D = \diag(d_1, \dots, d_n)$
with $d_1, \dots, d_n \in (0, \infty)$, and let $S$ be an
$n\times n$ column stochastic
matrix such that $RSR^{-1}$ is symmetric for some diagonal matrix $R$.
For $t \in [0, 1]$, denote by $\varrho(F(t))$ the Perron (largest) 
eigenvalue of 
$$F(t) = D[(1-t)I_n + tS]D^{-1}[(1-t)I_n + tS].
$$
Then $\varrho(F(t))$ is either an increasing function on $[0,1]$
or a constant function on $[0,1]$.
The latter case happens if and only if $D$ and $S$ commute, equivalently,
there is a permutation matrix $P$ such that 
$PSP^T = S_1 \oplus S_2 \oplus \cdots \oplus S_k$
and $PDP^T = d_1 I_{n_1}  \oplus d_2 I_{n_2} \oplus \cdots \oplus d_k 
I_{n_k}$, where $S_j \in M_{n_j}$ for $j = 1, \dots, k$,
and $n_1 + \cdots + n_k = n$. 
\end{theorem}

\bigskip\noindent
{\bf Remark.} The above result covers the case of symmetric $S$. It also covers 
the case when $S$ is an irreducible tridiagonal column stochastic matrix; one 
can use a simple continuity argument to extend the result to reducible 
tridiagonal column stochastic matrices. 

To prove the first assertion of Theorem~\ref{thm:main}, assume that  
$\max_id_i=0$. Assumption \textbf{A2} implies that 
$\bF(2,\hat\bx(2))\bF(1,\hat\bx(1))=I$. Theorem~\ref{thm:1} implies 
$\cI(\bd;\tilde d)>1$ for all $\tilde d\in (0,1]$. 

To prove the second assertion of Theorem~\ref{thm:main}, assume that $\bd$ is a (possibly mixed) Nash equilibrium. The first assertion of the Theorem implies that $\max_i d_i >0$.  Next, suppose to the contrary that there exists $j$ such that
$ f^j(1,\|\hat \bx^j(1)\|)f^j(2,\|\hat \bx^j(2)\|)>1$ 
i.e. there is a source patch. Then 
\[
\cI(\bd;0)=\max_j \sqrt{ f^j(1,\|\hat \bx^j(1)\|)f^j(2,\|\hat \bx^j(2)\|)}>1
\]
and by continuity $\cI(\bd;\tilde d)>1$ for all $\tilde d\ge 0$ sufficiently small. As this 
contradicts the assumption that $\bd$ is a Nash coalition, it follows that  
$f^j(1,\|\hat \bx^j(1)\|)f^j(2,\|\hat \bx^j(2)\|)\le 1$ for all $j$.  Finally, suppose to the contrary that $f^j(1,\|\hat \bx^j(1)\|)f^j(2,\|
\hat \bx^j(2)\|)= 1$ 
for all $j$. Theorem~\ref{thm:1} implies that $\cI(\bd;\max_id_i)>1$ 
which contradicts the fact that $\cI(\bd;\max_id_i)=1$.

\section{Proof of Theorem~\ref{thm:1}}

Denote by $\|A\|$ the operator norm of the matrix $A$.
The proof of Theorem \ref{thm:1} depends on the following.

\begin{theorem} \label{thm:2} Suppose $A \in M_n$ is nonzero and 
satisfies $\|I+ A\| \ge 1$. Then 
$\|I+tA\| \ge \|I+A\|$ for all $t \ge 1$ and one of the following
condition holds.

\medskip
{\rm (a)}  The function 
$t \mapsto \|I+tA\|$ is increasing for $t \ge 1$.

\medskip

{\rm (b)}  
There is a unitary matrix $U$ such that
$U^*AU = 0_k \oplus \tilde A$, where  
$\tilde A \in M_{n-k}$ is invertible
and satisfies $\|I_{n-k} + \tilde A\| < 1$.
Consequently, there is $t^* > 1$ such that 
$\|I_{n-k}+t^*\tilde A\| = 1$ and
the function $t \mapsto \|I+tA\|$
has constant value 1 for $t \in [1, t^*]$
and is increasing for $t > t^*$.
\end{theorem}

\it Proof. \rm Let $u$ be a unit vector 
such that $\|I+A\| = \|(I+A)u\|$ and 
$Au = \alpha u + \beta v$, where $\{u, v\}$ is an orthonormal family.
By our assumption, 
$$\|(I+A)u\| = \|(1+\alpha) u + \beta v\| = |1+\alpha|^2 + |\beta|^2 \ge 1,$$
i.e., 
$$2{\rm Re}(\alpha) + |\alpha|^2 + |\beta|^2 \ge 0.$$
Thus, for $t > 1$, 
\begin{eqnarray*}
\|I+tA\| 
&\ge& \|(I+tA)u\| \\
&=&  |1+ t\alpha|^2 + |t\beta|^2 \\
&=& |1 + \alpha|^2 + |\beta|^2 + 2(t-1) {\rm Re}(\alpha) + 
(t^2-1)[|\alpha|^2+|\beta|^2] \\
&=& \|I+A\| + (t-1)[2{\rm Re}(\alpha)+|\alpha|^2+|\beta|^2]
+ t(t-1)[|\alpha|^2+ |\beta|^2] \\
&\ge& \|I+A\|.
\end{eqnarray*}

(a) Suppose there is unit vector $u$ in the above calculation 
satisfying $Au \ne 0$, i.e., $(\alpha, \beta) \ne (0,0)$.
Then the last inequality in the  calculation is a strict 
inequality.  Thus, $\|I+tA\| > \|I+A\|$.

Now, for any $t_0 > 1$, we have $\|I + t_0 A\| > \|I+A\| \ge 1$.
If $u_0$ is a unit vector satisfying $\|(I+t_0 A)u_0\| = \|I + t_0A\|$,
then $t_0 A u_0 \ne 0$; otherwise, $\|I+t_0A\| = 1$. 
Consequently, if we replace  $A$ by $t_0A$ in the above proof,
we have
$$\|I+t(t_0A)\| > \|I+t_0A\|$$ 
for any $t \ge 1$.
Thus the function $t \mapsto \|I+tA\|$ is increasing for $t > 1$.

(b) Suppose $Au = 0$ for every unit vector $u$ satisfying
$\|(I+A)u\| = \|I+A\|$.  In particular, we have
$$\|I+A\| = \|(I+A)u\| = \|u\| = 1.$$
Let $U$ be unitary such that 
$U^*AU$ is lower triangular form with the first 
$k$ diagonal entries equal to zero, and the last $n-k$
diagonal entries nonzero. 
Since 
$$\|I + A\| =  \|I + U^*AU\| = 1,$$
we see that 
$$e_j^*(I+U^*AU)^*(I+U^*AU)e_j = \|(I+U^*AU)e_j\|^2 \le 1$$
for $j = 1, \dots, k$. (Here 
$\{e_1, \dots, e_n\}$ is the standard basis for $\IC^n$.)
As a result, we see that $U^*AU = 0_k \oplus \tilde A$,
where $\tilde A \in M_{n-k}$ is invertible.

Since $Au = 0$ for every unit vector $u$ satisfying 
$\|(I+A)u\| = \|I+A\| = 1$, we conclude that 
$\|(I_{n-k} + \tilde A)v\| < 1$ for any unit vector $v \in \IC^{n-k}$.
Thus, $\|I_{n-k} + \tilde A\| < 1$. 

Note that $A \ne 0$ implies $\tilde A$ is non-trivial, i.e., $k \ne n$. 
For sufficiently large $t$, we have $\|I_{n-k}+ t\tilde A\| \ge 1$.
Let $t^*$ be the smallest real number in $(1, \infty)$
such that $\|I_{n-k}+t^*\tilde A\| = 1$.
Since $\tilde A$ is invertible, case (a) must hold
and the function $t \mapsto \|I_{n-1} + t \tilde A\|$
is increasing for $t \ge t^*$. 
Hence for $t \ge t^*$, we have  $\|I + tA \| =  \|I_{n-k} + \tilde A\|$
and the function $t \mapsto \|I+tA\|$ is increasing.
\qed

\bigskip
\it Proof of Theorem \ref{thm:1}. \rm
Note that 
$F(t) = D[(1-t)I_n + tS]D^{-1}[(1-t)I_n + tS]$ 
and 
\begin{eqnarray*}
\tilde F(t) &=& RF(t)R^{-1}\\
&=& RDR^{-1}[(1-t)I_n + tRSR^{-1}]RD^{-1}R^{-1}[(1-t)I_n + tRSR^{-1}] \\
&=& D[(1-t)I_n + tRSR^{-1}]D^{-1}[(1-t)I_n + tRSR^{-1}]
\end{eqnarray*}
have the same eigenvalues and hence the same spectral radius.
Here, we use the fact that $RDR^{-1} = D$ as $R$ is 
a diagonal matrix.
So, $\tilde S = RSR^{-1}$ is symmetric with largest eigenvalue 
equal to 1, and all other eigenvalues lying in $[-1,1]$.
Moreover, if 
$$B(t) = D^{1/2}[(1-t)I_n + t\tilde S]D^{-1/2},$$
then
$D^{-1/2}F(t)D^{1/2} = B(t)B(t)^T$ so that 
$\varrho(F(t)) = \varrho\left(D^{-1/2}F(t)D^{1/2}\right) =  \|B(t)\|^2$.

Suppose $0 \le t_0 < t_0 + t \le 1$.
Let 
$A_0 = t_0 D^{1/2}(\tilde S-I_n)D^{-1/2}$. 
Then 
$$\|I+A_0\| = \|B(t_0)\| \ge \varrho(B(t_0)) = 1.$$
By Theorem \ref{thm:2},
\begin{equation}\label{eqn}
\|B(t_0+t)\| = \|I+ (1+t/t_0)A_0\| \ge \|I+A_0\| = \|B(t_0)\|.
\end{equation}
Thus, $\varrho(F(t)) = \|B(t)\|^2$ is non-decreasing on $[0, 1]$. 
Moreover, by Theorem \ref{thm:2}
the inequality in (\ref{eqn}) is an equality if 
and only if $A_0$ is unitarily similar to $0_k \oplus \tilde A_0$
where $\tilde A_0$ is invertible.
Thus, 

\bigskip
\ (1) the null space of $A_0$ has dimension $k$, and 

\ (2) $A_0v = 0$ if and only if $v^TA_0 = 0$.

\bigskip\noindent
From (1), we see that the eigenspace of the Perron root of the matrix
$D^{1/2}\tilde S D^{-1/2}$ has dimension $k$. So, the matrix
is permutationally similar to a $(k+1)\times (k+1)$ upper 
triangular block matrix so that each of the first $k$ 
diagonal block is square irreducible and has eigenvalue 1,
and the last block has eigenvalue less than 1.
Since $\tilde S$ is symmetric and $D$ is diagonal, there is a
permutation matrix $P$ such that $P\tilde SP^T
= \tilde S_1 \oplus \cdots \oplus \tilde S_{k+1}$ 
and $PDP^{-1} = D_1 \oplus \cdots \oplus D_{k+1}$
with $\tilde S_j, D_j \in M_{n_j}$ and $n_1 + \cdots + n_{k+1} = n$.
But then $D_{k+1}^{1/2}S_{k+1}D_{k+1}^{-1/2}$ cannot have 
spectral radius less than 1. So, $S_{k+1}$ must be 
vacuous.

Next, we show that each $D_j$ is a scalar matrix. To this end,
note that we can construct a null vector 
of $PA_0P^T$ by extending a Perron vector $u_1 \in \IR^{n_1}$ 
of $\tilde S_1$ to a vector $v_1$ in $\IR^n$ by adding $n-n_1$ zeros.
Then $PA_0P^Tv_1 = 0$. By (2), we see that $v_1^TPA_0P^T = 0$.
It follows that
$D_1^{1/2}\tilde S_1 D_1^{-1/2} u_1 = u_1$ and 
$u_1^T D_1^{1/2} \tilde S_1 D_1^{-1/2} = u_1^T$.
Equivalently, 
$\tilde S_1 D_1^{-1/2}u_1 = D_1^{-1/2}u_1$ and 
$u_1^T D_1^{1/2} \tilde S_1 = u_1^T D_1^{1/2}$.
Since $\tilde S_1$ is symmetric, nonnegative, and irreducible,
the eigenspace of the Perron eigenvalue is one dimensional
and there is a positive vector $w$ such that 
$\tilde S_1 w = w$ and $w^T \tilde S_1 = w^T$.
Hence $D_1^{-1/2}u_1$ is a multiple of $w$ and $u_1^TD_1^{1/2}$
is a multiple of $w^T$. Hence $D_1$ is a scalar matrix.
Similarly, one can prove that $D_2, \dots, D_k$ are scalar
matrices.  Since $P\tilde SP^T = \tilde S_1 \oplus \cdots \oplus \tilde S_k$,
we have $PSP^T = S_1 \oplus \cdots \oplus S_k$.

Conversely, if $S$ and $D$ commute, then
$F(t) = (tI + (1-t)S)(t I + (1-t)S)$ is column
stochastic for all $t \in [0, 1]$.
So, $\varrho(F(t)) = 1$ for all $t \in [0,1]$.
\qed

\medskip\noindent
{\bf Remark}
Note that the conclusion of Theorem \ref{thm:1} fails if $S$ is not 
diagonally similar to 
a symmetric matrix. For example, if $S = \begin{pmatrix}0 & 1 & 0 \cr 0 & 0 & 1\cr 1 & 0 & 0 \end{pmatrix}\oplus 
I_{n-3}$, and $D = \diag(3,2,1) \oplus I_{n-3}$, then $\varrho(F(0)) = \varrho(F(1)) = 1$, but 
$\varrho(F(1/2)) > 1$.
One may perturb $S$ slightly, say, replacing it by 
$(1-\varepsilon)S +  \varepsilon J$ for a small $\varepsilon > 0$
so that $\varrho(F(t))$ is not monotone, where $J$ is the matrix with all 
entries equal to $1/n$.

\textbf{Acknowledgements.}  We thank two anonymous reviewers and Fabio Dercole for their extensive and helpful comments. The research of the first author was partially supported by U.S. National Science Foundation Grants DMS-0517987 and EF-0928987. The research of the second author was partially
   supported by the U.S. National Science Foundation and the William and Mary Plumeri
   Award. He is an honorary professor of the
   University of Hong Kong and an honorary
   professor of the Taiyuan University of
   Technology.

\bibliography{../../seb}

\end{document}